\documentclass[10pt]{article}
\usepackage{amsmath, amsfonts}
\newcommand{\Tr}{\mathrm{Tr}}
\newcommand{\vecX}{{\vec X}}
\newcommand{\veca}{{\vec a}}

\newcommand{\veci}{{\vec i}}

\usepackage{fullpage}

\newtheorem{theorem}{Theorem}

\newtheorem{lemma}[theorem]{Lemma}

\newtheorem{proposition}[theorem]{Proposition}
\newtheorem{definition}[theorem]{Definition}
\newtheorem{claim}[theorem]{Claim}

\newcommand{\qed}{\rule{7pt}{7pt}}

\newenvironment{proof}{\noindent{\bf Proof}\hspace*{1em}}{\qed\bigskip}
\newenvironment{proof-sketch}{\noindent{\bf Sketch of Proof}\hspace*{1em}}{\qed\bigskip}
\newenvironment{proof-idea}{\noindent{\bf Proof Idea}\hspace*{1em}}{\qed\bigskip}
\newenvironment{proof-of-lemma}[1]{\noindent{\bf Proof of Lemma #1}\hspace*{1em}}{\qed\bigskip}
\newenvironment{proof-attempt}{\noindent{\bf Proof Attempt}\hspace*{1em}}{\qed\bigskip}
\newenvironment{proofof}[1]{\noindent{\bf Proof
of #1:}}{\qed\bigskip}

\makeatletter
\def\fnum@figure{{\bf Figure \thefigure}}
\def\fnum@table{{\bf Table \thetable}}
\long\def\@mycaption#1[#2]#3{\addcontentsline{\csname
  ext@#1\endcsname}{#1}{\protect\numberline{\csname
  the#1\endcsname}{\ignorespaces #2}}\par
  \begingroup
    \@parboxrestore
    \small
    \@makecaption{\csname fnum@#1\endcsname}{\ignorespaces #3}\par
  \endgroup}
\def\mycaption{\refstepcounter\@captype \@dblarg{\@mycaption\@captype}}
\makeatother

\newcommand{\mathify}[1]{\ifmmode{#1}\else\mbox{$#1$}\fi}
\newcommand{\bigO}O

\newcommand{\Z}{{\mathbb Z}}
\newcommand{\F}{{\mathbb F}}

\newcommand{\eqdef}{{\stackrel{\rm def}{=}}}

\newcommand{\poly}{{\rm {poly}}}

\newcommand{\wt}{{\rm {wt}}}
\renewcommand{\vec}[1]{{\mathbf #1}}

\def\fraka{{\mathfrak {a}}}

\def\blfootnote{\xdef\@thefnmark{}\@footnotetext}

\newcommand{\ii}{{\mathbf i}}
\newcommand{\bl}{{\mathbf e}}
\newcommand{\ba}{{\mathbf a}}
\newcommand{\bb}{{\mathbf b}}
\newcommand{\vecY}{{\mathbf Y}}
\newcommand{\mult}{\mathsf{mult}}
\newcommand{\Enc}{{\mathrm {Enc}}}
\newcommand{\bu}{{\mathbf u}}
\newcommand{\by}{{\mathbf y}}
\newcommand{\bz}{{\mathbf z}}

\begin{document}
\title{Some Remarks on Multiplicity Codes  }
\author{Swastik Kopparty\thanks{Department of Mathematics \& Department of Computer Science, Rutgers University. {\tt swastik.kopparty@rutgers.edu}. Research supported in part by a Sloan Fellowship and NSF CCF-1253886.}}
\date{October 20, 2013}
\maketitle

\centerline {\small {\em To Ilya Dumer, on the occasion of his 60th birthday  } } 

\begin{abstract}
Multiplicity codes are algebraic error-correcting codes generalizing
classical polynomial evaluation codes, and are based on evaluating polynomials
{\em and their derivatives}. This small augmentation confers upon them better
{\em local decoding}, {\em list-decoding} and {\em local list-decoding} algorithms
than their classical counterparts. We survey what is known about these codes,
present some variations and improvements, and finally
list some interesting open problems.
\end{abstract}

\newpage
\section{Introduction}

Reed-Solomon codes
and Reed-Muller codes are classical families
of error-correcting codes which have been widely influential in
coding theory, combinatorics and theoretical computer science.
These codes are based on evaluations of polynomials:
a codeword of one of these codes is obtained by evaluating a polynomial
over a finite field $\F_q$ of degree at most $d$ at all points in 
$\F_q^m$.

Multiplicity codes are a family of recently-introduced algebraic
error-correcting codes based on evaluations of polynomials and their
derivatives. Specifically, a codeword of a multiplicity code
is obtained by evaluating a polynomial of degree at most $d$, along
with all its derivatives of order $< s$, at all points in $\F_q^m$.

The $s = 1$ versions of multiplicity codes are thus the classical 
Reed-Solomon ($m=1$) and Reed-Muller ($m \geq 1$) codes. We will see that by allowing $s$ to be larger than $1$, in many senses general multiplicity codes go beyond their $s=1$ counterparts. 

Multiplicity codes with $m = 1$ (i.e., based on univariate polynomials)
were first considered by Rosenbloom and Tsfasman~\cite{RT-m-metric}, who studied
them for the purposes of producing optimal codes for the ``M metric" (now known as the Rosenbloom-Tsfasman metric). They were also studied by Nielsen~\cite{Nielsen}, 
who showed that they admit list-decoding algorithms upto the Johnson bound,
similar to the Reed-Solomon codes.

Multiplicity codes with general $m, s$ were defined by Kopparty, Saraf and Yekhanin~\cite{KSY}. The main result
of~\cite{KSY} was that for every $\epsilon, \alpha > 0$, for all $k$, there are multiplicity codes of dimension $k$,
rate $1 - \alpha$, and which are locally decodable from a constant fraction of errors with in just
$O_{\epsilon, \alpha}(k^{\epsilon})$ time. Prior to~\cite{KSY}, codes with nontrivial local decoding
algorithms were known only at rate $R < 1/2$, and achieving local decoding complexity
$O\left(k^{\epsilon}\right)$ required the code to have rather small rate $R = \epsilon^{(1/\epsilon)}$ (the codes that were known to achieve these parameters were the Reed-Muller codes). It should be noted that more recent results
have shown how to construct codes achieving parameters similar to those of multiplicity codes using significantly different ideas: Guo-Kopparty-Sudan~\cite{GKS}, Guo~\cite{Guo} and Hemenway-Ostrovsky-Wooters~\cite{HOW}.

Subsequently, Guruswami-Wang~\cite{GW} and Kopparty~\cite{K} studied the
list-decoding of univariate multiplicity codes, and showed that
there are sequences of univariate multiplicity codes of rate $R$, list-decodable
from $1 - R - \epsilon$ fraction errors in polynomial time
(achieving the so-called list-decoding capacity, thus providing another 
route to such codes after the breakthrough results of Parvaresh-Vardy~\cite{PV}
and Guruswami-Rudra~\cite{GR}). 

Global decoding of multivariate multiplicity codes was also considered
in~\cite{K}. There it was shown that multivariate multiplicity codes
can be decoded upto half their minimum distance in polynomial time,
and can be list-decoded from the Johnson bound in polynomial time.

The primary purpose of this paper is to survey the state of the art algorithms
for dealing with multiplicity codes. Along the way
we note some variations and improvements. Specifically:
\begin{enumerate}
\item We give an improved local decoding algorithm for multiplicity codes.
The original local decoding algorithm of~\cite{KSY} for multiplicity codes
worked as follows: in order to recover the correct value of the
multiplicity codeword at a point $\ba \in \F_q^m$, one would 
take $s^{O(m)}$ random lines in $\F_q^m$ passing through $\ba$,
query the codeword on all those lines, and use the answer to decode
the correct value at $\ba$. Our improved local decoding algorithm
is based on queries only $\exp(m)$ random lines through $\ba$.

This new algorithm is based on two new ideas.
First, we show that one can extract much more information
from each line about the correct value at $\ba$ than what the previous
algorithm took advantage of. Second, we use a more sophisticated
way of combining information from the different lines. For the
previous algorithm, the problem of combining information from the
various lines through
$\ba$ to recover the correct value of the codeword at $\ba$ amounted to the
problem of decoding a Reed-Muller code. In the new algorithm, this
problem turns out to be a case of decoding a multiplicity code! 

\item The above framework admits a number of variations that could
potentially be interesting for their own sake.

One variation leads to a ``polynomial rate" constant-query
error-correction scheme as follows:
a message $\sigma \in \Sigma_0^n$, where $|\Sigma_0| = \exp(n)$,
 gets encoded into a codeword $c \in \Sigma^n$, where
$\log |\Sigma| = n^{\epsilon} \cdot \log|\Sigma_0|$, such that even if a constant
fraction of the coordinates of $c$ are corrupted,
for any given\footnote{We use $[n]$ to denote the set $\{1, 2, \ldots, n\}$.}
$i \in [n]$ one can recover $\sigma_i$ with high
probability using only $O(1)$ queries into $c$. 
Such large alphabet error-correction schemes were considered by
Beimel and Ishai~\cite{BI01}.

%Note %that this code has ``polynomial rate".

Another variation allows local correction for some low
rate multiplicity codes using only $m$ lines, with a much
simpler local correction algorithm.

\item Using ideas from the above improvements, we give a new algorithm for 
(global) decoding of multivariate multiplicity codes.
The original approach of~\cite{K} was based on a family of
$s^{O(m)}$ space filling curves that passed through all the points
of $\F_q^m$. The new algorithm uses only $\exp(m)$ many curves.
The property of the $s^{O(m)}$ curves used in~\cite{K} was 
``algebraic repulsion":  no nonzero polynomial $P(X_1, \ldots, X_m)$
of moderate degree can vanish on all these curves.
The family of curves that we use in this paper can be smaller
because we require a weaker property: no nonzero polynomial
$P(X_1, \ldots, X_m)$ of moderate degree can vanish
on all these curves {\em with high multiplicity}.

\item We observe that encoding and 
unique decoding algorithms for multiplicity codes can be implemented in 
near-linear time (i.e., they run in time $O(n \cdot (\log n)^{O(1)})$) . For $m=1$, this follows from algorithms nearly identical
to the ones from the classical univariate ($s = m = 1$) case, and for
general $m$ it follows by refining a reduction to the $m = 1$ case given
in~\cite{K}.

\item We gather a number of open questions and possible future
research directions for the study of multiplicity codes.
\end{enumerate}

\paragraph{Organization of this paper:}
 In the next section we formally define multiplicity codes
and state their basic properties.
In Section~\ref{sec:decoding1} we discuss decoding algorithms for univariate multiplicity codes.
In Section~\ref{sec:decoding2} we discuss decoding algorithms for multivariate multiplicity codes.
 In Section~\ref{sec:encoding} we discuss encoding algorithms.
We conclude with some discussion and open questions.

\section{Multiplicity Codes}

We begin with some general preliminaries on codes, polynomials
and derivatives, and then move on to state the basic
definitions and results about multiplicity codes.

\subsection{Codes}

Let $\Sigma$ be a finite set and let $n$ be an integer. We will 
work with $\Sigma^n$ equipped with the (normalized) Hamming metric
$\Delta$, defined by:
$$ \Delta( x, y)  = \Pr_{i \in [n]} [ x_i \neq y_i ].$$

A \underline{code of length $n$ over the alphabet $\Sigma$} is a subset $\mathcal C$ of 
$\Sigma^n$. The rate of the code is defined to be:
$$ R = \frac{\log_{|\Sigma|} |\mathcal C|}{n}.$$
The minimum distance of the code $\mathcal C$ is defined to be the smallest
value $\delta$ of $\Delta( c, c')$ for distinct elements $c, c'$ of $\mathcal C$.

\paragraph{Encoding}
If $\mathcal C \subseteq \Sigma^n$ is a code, an \underline{encoding map}
for $\mathcal C$ is a bijection $E: \Sigma_0^k \to \mathcal C$ for some
integer $k$. Often $\Sigma_0 = \Sigma$, but it need not be.
It will be important that
this map $E$ is efficiently computable and efficiently invertible.

\paragraph{Unique Decoding}
In the problem of {\bf unique decoding} the code $\mathcal C$ from
$\eta$-fraction errors, where $\eta \leq \delta/2$, we are given as
input $r \in \Sigma^n$, and we wish to compute the unique $c \in \mathcal C$
(if any) such that $ \Delta(r,c) < \eta.$ The uniqueness
follows from our condition relating $\eta$ and $\delta$.

\paragraph{List-Decoding}
In the problem of {\bf list-decoding} the code $\mathcal C$ from $\eta$-fraction errors,
we are given as input $r \in \Sigma^n$, and we wish to compute the set
$$\mathcal L = \{ c \in \mathcal C \mid \Delta(r, c) < \eta \}.$$
The maximum possible value of $| \mathcal L|$ as $r$ varies over all elements of $\Sigma^n$
is called the {\bf list-size} for list-decoding $\mathcal C$ from $\eta$ fraction errors.

\paragraph{Local Correction and Local Decoding}

In the problem of {\bf locally correcting} the code $\mathcal C$ from
$\eta$-fraction errors, where $\eta \leq \delta/2$, we are given 
oracle access to a string $r \in \Sigma^n$, and given as input
$i \in [n]$, and we wish to compute $c_i$ for
the unique $c \in \mathcal C$ (if any) such that
$ \Delta(r,c) < \eta.$ The query complexity of such a local
correction algorithm is the number of queries made to $r$; both
the query complexity and time complexity could potentially be
sublinear in $n$ (and indeed this is the interesting case).

For {\bf local decoding}, we deal with a code $\mathcal C$
along with an encoding map $E: \Sigma_0^k \to \mathcal C$. In
the problem of local decoding $(\mathcal C, E)$
from $\eta$-fraction errors, where $\eta \leq \delta/2$,
we are given oracle access to $r \in \Sigma^n$, and input
$i \in [k]$, and we wish to compute $x_i$ for the unique
$x \in \Sigma_0^k$ (if any) such that $\Delta(r, E(x)) < \eta$.
The query complexity of such a local
decoding algorithm is the number of queries made to $r$; again,
both the query complexity and time complexity could potentially be
sublinear in $n$ (and indeed this is the interesting case).

The difference between local decoding and local correction is that
in local decoding, we are trying to recover symbols of the original message,
while in local correction, we are trying to recover symbols of the codeword.

\subsection{Polynomials and Derivatives}

For a vector $\veci = \langle i_1,\ldots,i_m\rangle$
of non-negative integers, its {\em weight}, denoted
$\wt(\veci)$, equals $\sum_{j=1}^m i_j$.

For a field $\F$, let $\F[X_1, \ldots, X_m] = \F[\vecX]$
be the ring of polynomials in the variables $X_1,\ldots,X_m$ with coefficients
in $\F$. For a vector of non-negative integers $\veci = \langle
i_1,\ldots,i_m \rangle$, let $\vecX^\veci$ denote the monomial
$\prod_{j= 1}^m X_{j}^{i_j} \in \F[\vecX]$.

We now define derivatives and the multiplicity of vanishing at a point.

\begin{definition}[(Hasse) Derivative]
For $P(\vec X) \in \F[\vec X]$ and
non-negative vector $\veci$, the $\veci$th {\em (Hasse) derivative} of
$P$, denoted $P^{(\vec i)}(\vec X)$,
is the coefficient of $\mathbf Z^{\veci}$ in the
polynomial $\tilde{P}(\vec X, \vec Z) \eqdef P(\vec X + \vec Z) \in \F[\vecX, \vec Z]$.
\end{definition}

Thus,
\begin{equation}
\label{eq:deridentity}
P(\vec X + \vec Z) = \sum_{\vec i} P^{(\veci)}(\vec X) \vec Z^{\vec i}.
\end{equation}

We will need some basic properties of the Hasse derivative (see~\cite{Hirschfieldbook}).

\begin{proposition}[Basic properties of Hasse derivatives]
\label{prop:prop}
Let $P(\vec X), Q(\vec X) \in \F[\vec X]^m$ and let $\vec i$, $\vec j$ be
vectors of nonnegative integers. Then:
\begin{enumerate}
%\item $P_d(\vecX) + Q_d(\vecX) = (P + Q)_d(\vecX)$.
\item $P^{(\veci)}(\vecX) + Q^{(\veci)}(\vecX) = (P+Q)^{(\veci)}(\vecX)$.
\item $(P \cdot Q)^{(\veci)}(\vecX) = \sum_{0 \leq \bl \leq \veci}P^{(\bl)}(\vecX) \cdot Q^{(\ii - \bl)}(\vecX)$.
\item $\left(P^{(\veci)}\right)^{(\vec j)}(\vecX) = {\veci + \vec j \choose \veci}P^{(\mathbf{i+j})}(\vecX)$.
\end{enumerate}
\end{proposition}

\begin{definition}[Multiplicity]
For $P(\vec X) \in \F[\vec X]$ and $\vec a \in \F^m$,
the {\em multiplicity} of $P$ at $\veca \in \F^m$,
denoted $\mult(P,\veca)$, is the largest integer $M$ such that
for every
non-negative vector $\veci$ with $\wt(\veci) < M$, we have
$P^{(\veci)}(\veca) = 0$
(if $M$ may be taken
arbitrarily large, we set $\mult(P, \veca) = \infty$).
\end{definition}

Next, we state a basic bound on the total
number of zeroes (counting multiplicity) that a polynomial can have on a product set $S^m$. An elementary proof of this lemma can be found in~\cite{DKSS}.

\begin{lemma}
\label{lemma:schwartz}
Let $P \in \F[\vecX]$ be a nonzero polynomial of total degree at most $d$.
Then for any finite $S \subseteq \F$,
$$\sum_{\veca \in S^m} \mult(P,\veca) \leq d\cdot |S|^{m-1}.$$
In particular, for any integer $s > 0$,
$$ \Pr_{\veca \in S^m}[\mult(P, \veca) \geq s] \leq \frac{d}{s|S|}.$$
\end{lemma}

\subsection{Multiplicity Codes}

Finally, we come to the definition of multiplicity codes.
\begin{definition}[Multiplicity code~\cite{KSY}]
Let $s,d,m$ be nonnegative integers and let $q$ be a prime power.
Let $\Sigma = \F_q^{m+s -1 \choose m} = \F_q^{\{\ii : \wt(\ii) < s\}}$.
For $P(X_1, \ldots, X_m) \in \F_q[X_1, \ldots, X_m]$, we define
{\em the order $s$ evaluation of $P$ at $\ba$}, denoted $P^{(< s)}(\ba)$,
to be the vector $\langle P^{(\ii)}(\ba)\rangle_{\wt(\ii) < s} \in \Sigma$.

\underline{The multiplicity code of order-$s$ evaluations of degree-$d$ polynomials in $m$ variables
over $\F_q$} is defined as follows.
The code is over the alphabet $\Sigma$, and has length $q^m$ (where the coordinates
are indexed by elements of $\F_q^m$). For each polynomial $P(\vecX) \in \F_q[X_1, \ldots, X_m]$ with
$\deg(P) \leq d$, there is a codeword in $\mathcal C$ given by:
$$\Enc_{s,d, m, q}(P) = \langle P^{(<s)}(\ba) \rangle_{\ba \in \F_q^m} \in (\Sigma)^{q^m}.$$
\end{definition}

Technically speaking, we have only defined the multiplicity
code as a subset of $\Sigma^{\F_q^m}$, without specifying an encoding map.
We postpone the choice of a good encoding map to a later section.

\begin{lemma}[Rate and distance of multiplicity codes~\cite{KSY}]\label{lem:multdist}
Let $\mathcal C$ be the multiplicity code of order $s$ evaluations
of degree $d$ polynomials in $m$ variables over $\F_q$.
Then $\mathcal C$ has minimum distance at least $\delta = 1-\frac{d}{sq}$ and
rate $\frac{{d+m \choose m}}{{s + m - 1 \choose m}q^m}$, which is at least
$$\left(\frac{s}{m+s}\right)^m \cdot \left(\frac{d}{sq}\right)^m
\geq \left(1- \frac{m^2}{s}\right)\left(1 - \delta\right)^m.$$
\end{lemma}

We usually think of $m$ and $s$ as large constants (significantly
smaller than $q$), and in light of the above parameters, having 
$s \gg m^2$ is particularly interesting. For the rest of this paper,
when we speak of near-linear time algorithms, this assumes that
$m$ and $s$ are constants, and that $q$ and the blocklength $q^m$ tend
to $\infty$.

One can easily convert such codes into codes over a constant sized (and even binary) alphabet via concatenation, while preserving the local decoding/correction properties. For details, see~\cite{KSY}.

\section{Decoding Univariate Multiplicity Codes}
\label{sec:decoding1}

We begin by discussing decoding of univariate multiplicity codes.

\subsection{Unique Decoding}

The classic Berlekamp-Welch algorithm for decoding Reed-Solomon codes
up to half the minimum distance has a simple generalization to 
the case of univariate multiplicity codes. This generalization was first
discovered by Nielsen~\cite{Nielsen}\footnote{Nielsen's theorem analyzes
the decoding radius in terms of the $m$-metric, and implies the decoding
algorithms for the Hamming metric considered here.}.
In fact, Nielsen showed how to do list-decoding  of univariate multiplicity codes, discussed next.

Let us set the problem up. Recall that the alphabet for this code is $\F_q^s$.
Thus the received word is a function $r: \F_q \to \F_q^{s}$. 
Abusing notation, we view this as a tuple of $s$ functions 
$r^{(i)}: \F_q \to \F_q$ for $0 \leq i < s$.
We wish to find
the unique $P(X)$ such that $\Delta( \Enc_{s,d,1,q}(P), r) < \delta/2$.

The algorithm tries to find an error-locator polynomial $E(X)$ and
another polynomial $N(X)$, such that $N(X) = E(X) \cdot P(X)$.
\begin{itemize}
\item Search for nonzero polynomials $E(X)$, $N(X)$ of degrees at most $(sq - d)/2$, $(sq + d)/2$ respectively such that
for each $x \in \F_q$, we have the following equations:
\begin{align}
N(x) &=  E(x) r^{(0)}(x)  \nonumber \\
N^{(1)}(x) &= E(x) r^{(1)}(x) + E^{(1)}(x) r^{(0)}(x)   \nonumber \\
\label{eqbw}
& \cdots  \\
N^{(s-1)}(x) &= \sum_{i=0}^{s-1} E^{(i)}(x) r^{(s-1-i)}(x) \nonumber
\end{align}
This is a collection of $sq$ homogeneous linear equations in $(sq-d)/2 + 1 + (sq + d)/2 + 1 > sq$ unknowns (the coefficients of $E$ and $N$).
Thus a nonzero solution $E(X), N(X)$ exists. Take any such nonzero solution.
\item Given $E(X)$, $N(X)$ as above, output $\frac{N(X)}{E(X)}$.
\end{itemize}

The analysis proceeds by showing that $N(X) - P(X) E(X)$, 
which is a degree $(sq + d)/2$ polynomial, has $> (sq + d)/2s$ 
zeroes of multiplicity $\geq s$, and is thus the zero polynomial.
This implies that $P(X) = N(X) / E(X)$, and so $P(X)$ is the output of
the algorithm, as desired.

\subsubsection{Unique decoding in near-linear time}

In this subsection we describe how to implement the above algorithm
in near-linear time. The presentation follows the description
of a near-linear time implementation of the Berlekamp-Welch algorithm
in Sudan's lecture notes~\cite{S-lintime}.

Let $R(X)$ be the unique polynomial of degree at most $sq-1$
such that for each $\alpha \in \F_q$ and $i < s$,
$$R^{(<s)}(\alpha) = r^{(<s)}(\alpha).$$
Such an $R(X)$ can be found in near-linear time 
by the classical Hermite interpolation algorithm of Chin~\cite{Chin}.
If $E(X)$ and $N(X)$ satisfy the equations~\eqref{eqbw},
then we have that $N(X) - E(X) R(X)$ vanishes at each
$x \in \F_q$ with multiplicity at least $s$.
Thus:
$$ N(X) = E(X) R(X) - C(X) \cdot (X^q - X)^s,$$
for some $C(X) \in \F_q[X]$.
Equivalently, 
$$ \frac{N(X)}{E(X)(X^q - X)^s} = \frac{R(X)}{(X^q - X)^s}  - \frac{C(X)}{E(X)}.$$
%Suppose we constrain $E(X)$ to have degree at most $D = (sq -d)/2$. Then
%we are seeking $C(X), E(X)$ so that $N(X)$ has degree at most $(sq + d)/2$.
%Precisely this problem can be solved in near-linear time via Strassen's
%continued fraction algorithm: given a rational function $\frac{R(X)}{Z(X)}$ and a degree bound 
%$D$, to find a rational function $\frac{C(X)}{E(X)}$ such that $\deg(E) \leq D$, and such that $\deg( R(X) E(X) - C(X) Z(X) )$ is minimized.

Thus we are looking for $C(X), E(X)$ such that:
\begin{enumerate}
\item $\deg(E(X)) \leq (sq-d)/2$,
\item the rational function $\frac{C(X)}{E(X)}$ approximates the rational function $\frac{R(X)}{(X^q - X)^s}$, in the sense that the numerator of their difference $N(X) = R(X) E(X) - C(X) (X^q - X)^s$
has degree at most $(sq+d)/2$.
\end{enumerate}
This problem can be solved in near-linear time via Strassen's
continued fraction algorithm~\cite{strassen-continued-fraction}. In fact, one can minimize the degree
of $N(X)$ subject to the constraint that $\deg(E(X)) \leq (sq- d)/2$.

Finally, the division step can also be performed in near-linear time.
This completes the description of the near-linear time implementation
of the unique decoder for univariate multiplicity codes.

\subsection{List-Decoding}

We now discuss the list-decoding of univariate multiplicity codes.
Here we consider the problem of decoding from a fraction of errors
which may be larger than half the minimum distance $\delta$.

By the Johnson bound, we know that for list-decoding
univariate multiplicity codes from $(1 - \sqrt{1 - \delta})$-fraction
errors, the list-size is at most $\poly(q)$ (this only uses
the fact that the distance of the code is $\geq \delta$).
It is thus reasonable to ask whether there is a polynomial time algorithm
to list-decode univariate multiplicity codes from  $(1 - \sqrt{1 - \delta})$-fraction error.

In~\cite{Nielsen}, Nielsen gave such an algorithm. His algorithm generalizes
the Guruswami-Sudan algorithm for list-decoding Reed-Solomon codes,
and is also based on  interpolation and root-finding.

Given a received word $r : \F_q \to \F_q^s$, one first
interpolates a low-degree bivariate polynomial $Q(X,Y) \in \F_q[X, Y]$
such that for each $\alpha \in \F_q$, the polynomial
$Q(X, \sum_{j=0}^{s-1} r^{(j)}(\alpha) (X-\alpha)^{j} )$
vanishes with high multiplicity at $X = \alpha$.
One then shows that every $P(X) \in \F_q[X]$ of degree at most $d$ with 
$\Delta( \Enc_{s,d,1,q}(P), r) \leq 1 - \sqrt{1-\delta}$, we have
$Q(X, P(X)) = 0$. Finally, one can find all polynomials
$P(X)$ satisfying this latter equation.

Recently Guruswami-Wang~\cite{GW} and Kopparty~\cite{K} independently found
improved results for list-decoding univariate multiplicity
codes over {\em prime fields}.

The main result of~\cite{GW}
is that order $s$ univariate multiplicity codes of distance $\delta$
over prime fields can, for every integer $0 \leq t < s$,
be list-decoded from $\eta_t$ fraction errors with list-size
at most $q^{O(s)}$, where:
$$\eta_t =  \frac{t+1}{t+2} \left(\delta - \frac{t}{s - t} \right).$$
For $t = 0$, the algorithm boils down to Nielsen's version of the Berlekamp-Welch
algorithm for unique-decoding multiplicity codes.

The main result of~\cite{K}
is that order $s$ univariate multiplicity codes of distance $\delta$
over prime fields can, for every integer $0 \leq t < s$
be list-decoded from $\eta'_t$ fraction errors with list-size
at most $q^{O(ts)}$, where:
$$\eta_t'= 1 - \left(\left(1 - \frac{t}{s-t}\right) \cdot (1-\delta) \right)^{\frac{t+1}{t+2}}.$$
For $t = 0$, the algorithm boils down to Nielsen's version of the Guruswami-Sudan algorithm for list-decoding univariate multiplicity codes.

Both these algorithms are based on
deriving an order $t$ differential equation of the form:
$$Q(X, P(X), P^{(1)}(X), \ldots, P^{(t-1)}(X)) = 0$$
from the received
word $r$, such that every $P$ whose encoding is close to $r$
must satisfy this differential equation. In the algorithm of~\cite{GW}
this differential equation is a linear differential equation,
and in the algorithm of~\cite{K} this equation is a polynomial differential
equation. These differential equations are then solved using
Hensel-lifting / power series. See~\cite{GW} and~\cite{K} for the details.
The decoding radius $\eta_r'$ is always greater than $\eta_r$,
but the algorithm and analysis of~\cite{K} are also more involved than 
that of~\cite{GW}.

It is well known that the maximimum fraction of errors $\eta$
from which a code of rate $R$ and block-length $n$
can be list-decoded from while still having  $\poly(n)$ list-size
is $1 - R -\epsilon$ (for arbitrarily small $\epsilon > 0$).
A code which achieves this is said to achieve list-decoding capacity.
The first constructions of codes which achieved list-decoding capacity
came from the breakthrough results of Parvaresh-Vardy~\cite{PV} and
Guruswami-Rudra~\cite{GR}. 
The above-mentioned results of~\cite{GW} and~\cite{K}
show that univariate multiplicity codes
over prime fields achieve list-decoding capacity for every $R \in (0,1)$.
This follows by noting that for univariate multiplicity codes, 
$R = 1 - \delta$, and that for every $\delta$, if we take $r$ to be
a very large constant, and $s$ to be a much larger constant, then the
above decoding radii $\eta_r$ and $\eta_{r'}$ approach
$\delta = 1 - R$.

\section{Decoding Multivariate Multiplicity Codes}
\label{sec:decoding2}

\subsection{Local Correction}

We begin by discussing local correction algorithms for multiplicity codes.
When coupled with a {\em systematic } encoding map (which we discuss in the next section), this also gives local decoding algorithms for multiplicity codes.

\subsubsection{Preliminaries on Restrictions and derivatives}

We first consider the relationship between the derivatives of a
multivariate polynomial $P$ and its restrictions to a line.
Fix $\ba, \bb \in \F_q^m$, and consider the polynomial $Q(T) = P(\ba + \bb T)$.
\begin{itemize}
\item
{\bf The relationship of $Q(T)$ with the derivatives of $P$ at $\ba$:}
By the definition of Hasse derivatives,
$$Q(T) = \sum_{\ii} P^{(\ii)}(\ba) \bb^\ii T^{\wt(\ii)}.$$
Grouping terms, we see that:
\begin{align}
\label{eq:restrict-derive}
\sum_{\ii \mid \wt(\ii) = j} P^{(\ii)}(\ba) \bb^{\ii} = \mbox{coefficient of $T^j$ in $Q(T)$.}
\end{align}
\item
{\bf The relationship of the derivatives of $Q$ at $t$ with the derivatives of $P$ at $\ba + t \bb$:} Let $t \in \F_q$.
By the definition of Hasse derivatives, we get the following two identities:
$$P(\ba + \bb(t+R)) = Q(t + R) = \sum_j Q^{(j)}(t) R^j.$$
$$P(\ba+\bb(t + R)) = \sum_{\ii} P^{(\ii)}(\ba + \bb t)(\bb R)^{\ii}.$$
Thus,
\begin{align}
\label{eq:restrict}
Q^{(j)}(t) = \sum_{\ii \mid \wt(\ii) = j} P^{(\ii)}(\ba + \bb t) \bb^{\ii}.
\end{align}
In particular, $Q^{(j)}(t)$ is simply a linear combination of the various
$P^{(\ii)}(\ba + \bb t)$ (over different $\ii$).
\end{itemize}

We now apply these observations to the derivatives of $P$.
For each nonnegative tuple $\bl \in \Z^m$, consider the polynomial $Q_{\bl}(T) = P^{(\bl)}(\ba + \bb T)$.
\begin{itemize}
\item
{\bf The relationship of $Q_{\bl}(T)$ with the derivatives of $P$ at $\ba$:}
\begin{align}
\label{eq:restrict-derive-bl}
\sum_{\ii \mid \wt(\ii) = j} (P^{(\bl)})^{(\ii)}(\ba) \bb^{\ii} = \sum_{\ii \mid \wt(\ii) = j} {\bl + \ii \choose \bl } P^{(\bl + \ii)}(\ba) \bb^{\ii} = \mbox{coefficient of $T^j$ in $Q_{\bl}(T)$.}
\end{align}
In particular, knowing $Q_{\bl}(T)$ gives us several linear relations between the evaluations of the derivatives of $P$ at $\ba$.
\item
{\bf The relationship of the derivatives of $Q_{\bl}$ at $t$ with the derivatives of $P$ at $\ba + t \bb$:} Let $t \in \F_q$.
We get
\begin{align}
\label{eq:restrict-bl}
Q_{\bl}^{(j)}(t) = \sum_{\ii \mid \wt(\ii) = j} (P^{(\bl)})^{(\ii)}(\ba + \bb t) \bb^{\ii} =  \sum_{\ii \mid \wt(\ii) = j} {\bl + \ii \choose \bl} P^{(\bl +\ii)}(\ba + \bb t) \bb^{\ii}  .
\end{align}
In particular, $Q_{\bl}^{(j)}(t)$ is simply a linear combination of evaluations, at $\ba + \bb t$, of the various derivatives of
$P$.
\end{itemize}

\subsubsection{The Local Correction Algorithm}
\label{sec:local}

We now give our local correction algorithm which corrects
$\delta_0 < \frac{\delta}{8}$ fraction errors.
The $\gamma= 0$, $c = 1$ case of this algorithm
is the orignal local correction algorithm of~\cite{KSY}. Increasing $\gamma$
reduces the query complexity from $s^{O(m)}$ to $\exp(m)$, while
reducing the fraction of correctable errors by a negligible amount.

\noindent{\bf Main Local Correction Algorithm:}\\
\noindent{\bf Input:} received word $r: \F_q^m \rightarrow \Sigma$, point $\ba \in \F_q^m$.
Abusing notation again, we will write $r^{(\ii)}(\ba)$ when we mean the $\ii$ coordinate of $r(\ba)$.
\begin{enumerate}
\item Set $\gamma = 1 - \frac{(1-\delta)}{1 - 8\delta_0} = \frac{\delta - 8\delta_0}{1 - 8\delta_0}$. Set $c = \gamma \cdot s + 1.$

%(If $\delta = 0.99$, $\delta_0 = 0.1$, then $\gamma = 0.98$.)

\item  {\bf Pick a set $B$ of directions:} 
Pick $\bz, \by_1, \by_2, \ldots \by_m \in \F_q^m$ independently and uniformly at random. Let $S \subset \F_q$ be any set of size $\lceil\frac{5s}{c}\rceil$.
Define
$$ B = \{ \bz + \sum_{j=1}^m \alpha_j \by_j \mid \alpha_j \in S \}.$$

%Via Lemma~\ref{lem:sample}, sample a set
%$B \subseteq \F_q^m$ of size at most $M = \max(\left(\lceil\frac{1}{0.97 \gamma}\rceil\right)^m, 100)$.

%with the desired well-sampling and general position properties.

\item {\bf Recover $P^{(\bl)}(\ba + \bb T)$ for directions $\bb \in B$:} For each $\bl$ with $\wt(\bl) < c$ and each $\bb \in B$, consider the function $\ell_{\bb, \bl} : \F_q \rightarrow \F_q^{s - \wt(\bl)}$ given by
\begin{align}
\label{eq:decode1}
(\ell_{\bb, \bl}(t))_j =  \sum_{\ii \mid \wt(\ii) = j} {\bl + \ii \choose \bl } r^{(\bl + \ii)}(\ba+\bb t) \bb^\ii,
\end{align}
for each $0 \leq j < s - \wt(\bl)$.
Via a univariate multiplicity code decoding algorithm,
find the unique polynomial $Q_{\bb, \bl}(T) \in \F_q[T]$ of degree at most $d-\wt(\bl)$ (if any), such that
$$\Delta(\Enc_{s - \wt(\bl), d - \wt(\bl), 1, q}(Q_{\bb, \bl}), \ell_{\bb, \bl}) < 2 \delta_0.$$

\item {\bf Decode a constant degree multiplicity code to recover $P^{(<s)}(\ba)$:}
Denote the coefficient of $T^j$ in $Q_{\bb, \bl}(T)$ by $v_{j, \bb, \bl} \in \F_q$. If $j < 0$, we define $v_{j, \bb, \bl} = 0$.

For each $j'$ with $0 \leq j' < s$, find the unique homogeneous 
degree $j'$ polynomial $R_{j'}(\vecX) \in \F_q[\vecX]$ such that
for at least $1/3$ of the $\bb \in B$, for all $\bl$ with $\wt(\bl) < c$,
we have:
$$ R_{j'}^{(\bl)}(\bb) = v_{j' - \wt(\bl), \bb, \bl}.$$
Note that this is a constant degree multiplicity code decoding problem.

If such an $R_{j'}$ does not exist, or is not unique, the algorithm outputs FAIL.

For each $\ii$ with $\wt(\ii) < s$, define $u_{\ii}$ to equal the coefficient
of $\vecX^{\ii}$ in $R_{\wt(\ii)}(\vecX)$.

%For each $j$ with $0 \leq j < s$, consider the following system of equations in the variables
%$\langle u_{\ii}\rangle_{\wt(\ii) = j}$ (with one equation for each $\bb \in B$ and $\bl$ with
%$\wt(\bl) < c$):
%with each $u_{\ii} \in \F_q$:
%\begin{align}
%\label{eq:system2}
%\sum_{\ii \mid \wt(\ii) = j} {\bl  + \ii \choose \bl} \bb^{\ii} u_{\bl + \ii } = v_{j, \bb, \bl}. \quad\quad (\star-j-\bb-\bl)
%\end{align}
%Find all $\langle u_{\ii} \rangle_{\wt(\ii) = j}$, such that for at least $3/5$
%of the $\bb \in B$, for all $\bl$ with $\wt(\bl) < c$, the equation $(\star-j-\bb-\bl)$ is satisfied.
%If there are $0$ or $>1$ such solutions, output FAIL.

\item Output the vector $\langle u_{\ii} \rangle_{\wt(\ii) < s}$.
\end{enumerate}

We quickly comment on the running time and query complexity.
The running time consists
of $|S|^m$ instances of decoding univariate multiplicity codes over $\F_q$,
as well as on instance of decoding a degree-$s$ $m$-variate order-$c$
multivariate multiplicity code with evaluation points being $S^m$. 
Thus, if $m, s$ are constant, the running time is near-linear in $q$,
which is near-linear in $n^{1/m}$, where $n$ is the block-length of the code.
The query complexity is $|S|^m \cdot q$, which equals $(\frac{5}{\gamma})^m \cdot n^{1/m}$. For $\delta = \Omega(1)$ and $\delta_0 < \delta / 10$ (say),
the query complexity equals $\exp(m) \cdot n^{1/m}$.

\subsubsection{Analysis of the Local Correction Algorithm}

We now analyze the above local correction algorithm.
\begin{theorem}
Let $P(\vecX) \in \F_q[\vecX]$ be such that $\Delta(\Enc_{s,d,m,q}(P), r) < \delta_0.$
Let $\ba \in \F_q^m$.

With high probability, the local correction algorithm above outputs $P^{(<s)}(\ba)$.
\end{theorem}

\begin{proof}
Let $E = \{x \in \F_q^m \mid P^{(<s)}(x) \neq r^{(<s)}(x) \}$ be the error set.
We have $|E| < \delta_0 \cdot q^m$.

Let $L_{\bb} =  \{\ba + t \bb \mid t \in \F_q \}$ be the line through $\ba$ in direction $\bb$.
We call $\bb$ bad if $|L_{\bb} \cap E| \geq 4 \cdot \delta_0 \cdot q.$

Note that at most $1/4$ of all the lines are bad.
\begin{claim}
With high probability, we have:
\begin{enumerate}
\item at most $1/3$ of the $\bb \in B$ are bad,
\item $|B| = |S|^m$,
%\item If $R(\vecX) \in \F_q[\vecX]$, with $\deg(R) < s$, satisfies
%$$\left| \{\bb \in B \mid \mult(R, \bb) \geq c\} \right| \geq \frac{1}{3} |S|^m,$$
%then $R(\vecX) = 0$.
\end{enumerate}
\end{claim}
These basic probability/linear-algebra facts are well known, and we omit the proofs.

% follows from Item 2 and the fact that $B$ is an
%affine image of the set $S^m$, and if $R'(\vecX)$ vanishes on 
%at least $1/3$ fraction of the points 

Henceforth we assume that both these events happen.

\begin{claim}
\label{goodq}
If $\bb$ is good, then for every $\bl$ with $\wt(\bl) < c$, we have:
$$Q_{\bb, \bl}(T) = P^{(\bl)}(\ba + \bb T).$$
\end{claim}
\begin{proof}
The univariate multiplicity code of order $s-\wt(\bl)$ evaluations of degree $d - \wt(\bl)$ polynomials has 
minimum distance at least $1 - \frac{d}{(s-c+1) q} = 1 -\frac{1-\delta}{1 -\gamma}$ which, by choice of $\gamma$, is $\geq 8 \cdot \delta_0$.

If $\bb$ is good, then we know that
$|L_{\bb} \cap E| < 4 \cdot \delta_0 \cdot q$.
By Equations~\eqref{eq:decode1} and~\eqref{eq:restrict-bl}, we conclude that
$P^{(\bl)}(\ba + \bb T)$ (which has degree $d - \wt(\bl)$) satisfies:
$$\Delta(\Enc_{s-\wt(\bl), d-\wt(\bl), 1, q}(P^{(\bl)}(\ba + \bb T)), \ell_{\bb, \bl}) \leq \frac{|L_{\bb} \cap E|}{q} < 4 \cdot \delta_0,$$
which is less than half the minimum distance of the  univariate multiplicity code of order $s-\wt(\bl)$ evaluations of degree $d - \wt(\bl)$ polynomials.

Thus $P^{(\bl)}(\ba + \bb T)$ is the unique such polynomial found in Step 3, and so $Q_{\bb, \bl}(T) = P^{(\bl)}(\ba + \bb T)$.
\end{proof}

For each integer $0 \leq j' < s$, define the polynomial:
$$\tilde{R}_{j'}(\vecX) = \sum_{\ii' \mid \wt(\ii') = j'} P^{(\ii')}(\ba) \vecX^{\ii'} .$$

\begin{claim}
If $\bb$ is good, then for all $\bl$ with $\wt(\bl) < c$, we have:
$$\tilde{R}_{j'}^{(\bl)}(\bb) = v_{j' - \wt(\bl), \bb, \bl},$$
\end{claim}
\begin{proof}
We have:
\begin{align*}
\tilde{R}_{j'}^{(\bl)} ( \vecX )
&= \sum_{\ii' \mid \wt(\ii') = j'} {\ii' \choose \bl}  P^{(\ii')}(\ba) \vecX^{\ii'-\bl}\\
&= \sum_{\ii' \mid \wt(\ii') = j', \ii' \geq \bl} {\ii' \choose \bl}  P^{(\ii')}(\ba) \vecX^{\ii'-\bl}\\
&= \sum_{\ii \mid \wt(\ii) = j} {\bl + \ii \choose \bl}  P^{(\bl + \ii)}(\ba) \vecX^{\ii},
\end{align*}
where $j = j' - \wt(\bl)$.

Thus,
\begin{align*}
\tilde{R}_{j'}^{(\bl)} ( \bb)
&= \sum_{\ii \mid \wt(\ii) = j} {\bl + \ii \choose \bl}  P^{(\bl + \ii)}(\ba) \bb^{\ii}\\
&= \mbox{coeff. of $T^j$ in $Q_{\bb, \bl}(T)$ \quad \quad (by Equation~\eqref{eq:restrict-derive-bl} and Claim~\ref{goodq}, since $\bb$ is good)}\\
&= v_{j' - \wt(\bl), \bb, \bl}.
\end{align*}
\end{proof}

Thus $\tilde{R}_{j'}(\vecX)$ satisfies the conditions required of Step 4 of the algorithm.

Let us now show that no other polynomial can satisfy these conditions.
Suppose there was some other solution $\overline{R}_{j'}(\vecX)$. Then the difference
$(\tilde{R}_{j'} - \overline{R}_{j'})(\vecX)$ would be a nonzero polynomial of degree $< s$, that
vanishes with multiplicity at least $c$, at $\geq \frac{1}{3}$ of the points of $B$.
But this cannot be, since $B$ is an affine one-to-one image of the set
$S^m$, and the fraction of points of $S^m$ on which a nonzero polynomial
of degree $<s$ can vanish with multiplicity $\geq c$ is at most 
$\frac{s}{c|S|} = \frac{1}{5} < \frac{1}{3}$.
Thus $\tilde{R}_{j'}$ is the unique solution found in Step 4.

Finally, we notice that our definition of $R_{j'}$ implies that 
for every $\ii$, we have $u_{\ii} = P^{(\ii)}(\ba)$, as desired.
\end{proof}

\subsubsection{Variations}
\label{sec:var}

The above algorithm allows a number of variations that may be useful in different
contexts.

Let $\ba \in \F_q^m$. Suppose $r : \F_q^m \to \Sigma$ is a received word, 
and suppose $P(\vecX) \in \F_q[\vecX]$ is a polynomial of degree at most
$d$ such that $\Delta( \Enc_{s,d,m,q}(P), r) < \delta_0$.
Let $\gamma = \frac{\delta - 8 \delta_0}{1-8\delta_0}$, and let 
$ c= \gamma s + 1$.

Let $\ba \in \F_q^m$.
For each integer $0 \leq j' < s$, define the polynomial:
$$\tilde{R}_{j'}(\vecX) = \sum_{\ii' \mid \wt(\ii') = j'} P^{(\ii')}(\ba) \vecX^{\ii'} .$$
Suppose $\bb \in \F_q^m$ is good (meaning that
the line $L_{\ba, \bb}$ has $< 4 \delta_0 q$ errors on it).
As we saw in the above analysis, by querying all the points of the 
line $L_{\ba, \bb}$, we can compute  $\tilde{R}_{j'}^{(\bl)}(\bb)$,
for every $j' < s$ and every $\bl$ such that $\wt(\bl) < c$.

\begin{enumerate}
\item Suppose we are only interested in recovering 
$P^{(< c)}(\ba)$. Then this can be recovered by querying the points of
{ just one line!} Indeed, if we pick $\bb$ at random,
then with high probability $\bb$ is good, and then by 
querying $L_{\ba, \bb}$ we can compute 
 $\tilde{R}_{j'}^{(\bl)}(\bb)$ 
for every $j' < s$ and every $\bl$ such that $\wt(\bl) < c$.
Note that for every $\bl$ with $\wt(\bl) < c$, we have
$\tilde{R}_{\wt(\bl)}^{(\bl)}(\bb) = P^{(\bl)}(\bb)$.
Thus we can compute $P^{(<c)}(\ba)$ with high probability.

We now describe a coding scheme taking advantage of this.
Specializing parameters, if we take $\delta = 1/2$, 
$\delta_0 = 1/100$, $s = \Omega(m^2)$, $d = (1-\delta) \cdot s \cdot q$,
 $n = q^m$, $\gamma = \frac{1}{5}$, and so $c > \frac{s}{5}$.
Let $\Sigma_0 = \F_q^{ \{ \ii \mid \wt(\ii) < c \} }$,
and let our space of messages be the space of all functions
$f: \F_q^m \to \Sigma_0$. 

Define the encoding of the message $f$ as follows:
find any polynomial $P(\vecX) \in \F_q[\vecX]$ of degree at most $(c + m)\cdot q$
such that for each $\ba \in \F_q^m$, we have
$P^{(< c)}(\ba) = f(\ba)$ (that such a polynomial $P$ exists
is an interpolatability statement; it follows from the arguments
in Appendix A of~\cite{K}). By choice of parameters,
$(c+m) \cdot q < d$. The encoding of 
$f$ is then defined to be $\Enc_{s,d,m,q}(P) \in \Sigma^{n}$.
Note that $\log|\Sigma| = {m+s-1 \choose m} \cdot \log q$, and 
$\log|\Sigma_0| = {m + c -1 \choose m} \cdot \log q \approx \gamma^m \cdot \log|\Sigma|$.
Thus this encoding blows up the bit-length of an alphabet symbol 
by a factor of $5^m$, which is at most a sublinear polynomial in $n$
for $q > 5$.

%Observe that $|\Sigma| = |\Sigma_0|^{1 + o(1)}$, and so this
%scheme has ``polynomial rate" (the proviso being that
%the alphabet sizes are different).

By the earlier discussion on local correction,
this coding scheme has the following interesting property:
given any $\ba \in \F_q^m$, given oracle access to some 
$r \in \Sigma^m$ which is $\delta_0$-close to the encoding of 
$f$, one can recover the value of $f(\ba)$ with high probability
using only $q$ queries into $r$. This scheme can be used
even for $q$ as small as $O(1)$!

\item We now describe another local correction algorithm
for multiplicity codes. This algorithm queries only $m$ lines, but
it only works for multiplicity codes of low rate.

To locally correct the value of $P^{(<s)} (\ba)$ for
$\ba \in \F_q^m$, given oracle access to $r : \F_q^m \to \Sigma$,
the algorithm works as follows. First pick a set $B$
of $m$ uniformly random directions from $\F_q^m$. With high probability,
we will have that $B$ will be a set of $m$ linearly independent vectors,
and if the fraction of errors $\delta_0$ is sufficiently small, then
all the $\bb \in B$ will be good.
This means that we can compute $\tilde{R}_{j'}^{(\bl)}(\bb)$
for each $\bb \in B$, $j' < s$ and $\bl$ with $\wt(\bl) < c$.

The following lemma implies that this data uniquely determines
the polynomial $\tilde{R}_{j'}(\vecX)$ for each $j' < c' = c \cdot \frac{m}{m-1}$.

\begin{lemma}
\label{lem:joints1}
Let $c' = c \cdot \frac{m}{m-1}$.
Let $R(X_1, \ldots, X_m)$ be a homogeneous polynomial of degree $j' < c'$.
Suppose $B$ is a set of $m$ linearly independent vectors in $\F_q^m$,
such that $\mult(R, \bb) \geq c$ for each $\bb \in B$.

Then $R(X_1, \ldots, X_m) = 0$.
\end{lemma}
\begin{proof}
Without loss of generality, we may assume $B = \{\bu_1, \ldots, \bu_m\}$,
where  $\bu_i \in \F_q^m$ is $1$ in coordinate $i$ and $0$ in every
other coordinate.  The hypothesis $\mult(R, \bu_i) \geq c$ implies
that for every $\ii$ with $\wt(\ii) = j'$ and $\ii_i > j' - c$,
the coefficient of $\vecX^{\ii}$ in $R(\vecX)$ is $0$.

Finally, notice that for every $\ii$ with $\wt(\ii) = j'$,
there always exists some $i \in [m]$ for which $\ii_i > \frac{j'}{m} > j'-c$.

Thus $R(\vecX) = 0$.
\end{proof}

Once we have computed  $\tilde{R}_{j'}(\vecX)$ for each $j' < c' = c \cdot \frac{m}{m-1}$,
this immediately gives us $P^{(< c')}(\ba)$. If $c' \geq s$, 
then this is proper local correction algorithm, but even if $c' < s$ this
algorithm could be of interest.

For completeness, we record the basic multiplicity amplification fact underlying the above decoding algorithm.
This can also be proved using Lemma~\ref{lem:joints1}.

\begin{lemma}
\label{lem:joints2}
Let $c' = c \cdot \frac{m}{m-1}$.

Let $\ba \in \F^m$.
Let $P(X_1, \ldots, X_n) \in \F[X_1, \ldots, X_m]$, and let $B \subseteq \F_q^m$ be a basis for $\F_q^m$ over $\F_q$.
For each $\bb \in B$ and each $m$-tuple $\bl$ of nonnegative integers,
define $Q_{\bb, \bl}(T) = P^{(\bl)}(\ba + \bb T) \in \F[T]$.

Suppose $Q_{\bb, \bl}(T) = 0$ for each $\bb \in B$ and each $\bl$ with $\wt(\bl) < c$.
Then $\mult(P, \ba) \geq c'$.
\end{lemma}

\end{enumerate}

\subsection{Global Decoding}

We now consider decoding of multivariate multiplicity codes
in the global sense. In the case of standard polynomial codes
(the $s = 1$ case), the best known algorithm (due
to Pellikaan-Wu~\cite{PW}) for decoding $m$-variate codes over $\F_q$ works
via a reduction to the decoding of $1$-variate codes over the
bigger field $\F_{q^m}$. 

For multiplicity codes with general $s$, Kopparty~\cite{K} gave a reduction from
$m$-variate codes over $\F_q$ to {\em several instances} of decoding a $1$-variate code over the big field $\F_{q^m}$.
Using the algorithms for decoding univariate multiplicity codes discussed earlier,
this gives polynomial time algorithms for unique decoding multivariate multiplicity
codes upto half the minimum distance, and list-decoding multivariate multiplicity
codes upto the Johnson bound.

Below we describe a variation of the reduction of~\cite{K}.
The key ingredient of that reduction is the construction of a certain special
family of ``algebraically-repelling" curves.

Abusing notation, we call $\fraka \in \F_{q^m}^m$ a basis if its 
$m$ coordinates form a basis for $\F_{q^m}$ over $\F_q$.
To every basis $\fraka = (a_1, \ldots, a_m)$, we associate a curve
$\gamma_{\fraka}(T) \in \F_{q^m}[T]^m$, given by:
$$\gamma_{\fraka}(T) = (\Tr(a_1 T), \Tr(a_2 T), \ldots, \Tr(a_m T)).$$
The most interesting feature of this curve is that
$\gamma_{\fraka}$ is a bijection between $\F_{q^m}$ and $\F_q^m$.
See~\cite{K} for more properties of these curves $\gamma_{\fraka}$.

A collection $\fraka_1, \ldots, \fraka_M \in \F_{q^m}^m$ of bases
is said to be in $(s, c)$-general position if there does not
exist a nonzero polynomial $R(\vecX) \in \F_{q^m}[\vecX]$ of degree at most
$s$ which vanishes at each $\fraka_i$ with multiplicity at least $c$.

The $c = 1$ case of the following lemma was shown in~\cite{K}.
\begin{lemma}
\label{lem:zerozero}
Suppose $c \leq s < q$.
Let $\fraka_1, \ldots, \fraka_M \in \F_{q^m}^m$ be bases in $(s,c)$ general position.
Let $Q(\vecX) \in \F_{q^m}[\vecX]$ have degree $< s \cdot q$.
Suppose that for each $i \in [M]$ and each $\bl$ with $\wt(\bl) < c$,
the univariate polynomial
$Q^{(\bl)} \circ \gamma_{\fraka_i}(T) = 0$.

Then $Q(\vecX) = 0$.
\end{lemma}
The proof of this lemma is postponed to Section~\ref{sec:repel}.

Explicit collections $\fraka_1, \ldots, \fraka_M \in \F_{q^m}^m$
in $(s,c)$ general position with $M = \left(\frac{s}{c}\right)^m$
can be constructed as follows.
Take $\fraka$ to be a basis of $\F_{q^m}^m$, and 
$\bb_1, \ldots, \bb_M$ to be all the 
elements of a $m$-dimensional grid of side $\frac{s}{c}$ in $\F_q^m$,
and set $\fraka_i = \fraka + \bb_i$.

Via the lemma, this gives us an explicit collection of $\left(\frac{s}{c} \right)^m$
``algebraically-repelling" curves. We now show how these curves 
can be used for reducing multivariate decoding to univariate decoding.
Again, this generalizes the $\gamma = 0$, $c = 1$ case, which was done in~\cite{K}.

The following lemma, relating the derivatives of a multivariate polynomial to the derivatives of its restriction to the curve $\gamma_{\fraka}$, will motivate one of the steps of the algorithm.
\begin{lemma}
\label{lem:multcompose}
Let $P(X_1, \ldots, X_m) \in \F_{q}[X_1, \ldots, X_m]$ and let $Q_{\bl}(T) \in \F_{q^m}[T]$ be given by $Q_{\bl}(T) = P^{(\bl)} \circ \gamma_{\fraka}(T)$.
Then for every $t \in \F_{q^m}$ and every $j < q$:
$$ Q_{ \bl}^{(j)}(t) = \sum_{\ii : \wt(\ii) = j} {\ii + \bl \choose \ii} P^{(\ii+\bl)}(\gamma_{\fraka}(t)) \fraka^{\ii}.$$
\end{lemma}

{\bf Algorithm for Reducing Multivariate Decoding to Univariate Decoding}
\begin{enumerate}
\item Suppose we have an algorithm $\mathcal A$ that list-decodes
univariate multiplicity codes of distance $\delta$ from $\eta(\delta)$-fraction
errors.
\item Let $\gamma = \frac{\delta- \eta^{-1}(\delta_0)}{1- \eta^{-1}(\delta_0)}$, and let $c = \gamma \cdot s + 1$.
\item Let $M = (\frac{s}{c})^m$.
Pick bases $\fraka_1, \fraka_2, \ldots, \fraka_M \in \F_{q^m}^m$ in $(s,c)$-general position.
\item For each $i \in [M]$, for each $\bl$ with $\wt(\bl) < c$, define $\ell_{i,\bl} : \F_{q^m} \rightarrow \F_{q^m}^{s-c+1}$ as follows. For each $j$ with $0 \leq j < s-c+1$, let 
$$ (\ell_{i, \bl}(t))_j = \sum_{\ii: \wt(\ii) = j} {\ii + \bl \choose \bl} r^{(\ii + \bl)}(\gamma_{\fraka_{i}}(t)) \cdot \fraka_i^{\ii}.$$
\item Using algorithm $\mathcal A$, compute the set $\mathcal L_{i, \bl}$ of all $Q(T) \in \F_{q^m}[T]$ of degree at most $dq^{m-1}$ such that $\Delta( \Enc_{s-c+1,dq^{m-1},1,q^m}(Q), \ell_{i,\bl}) < \delta_0$.
\item For every $(Q_1(T), Q_2(T), \ldots, Q_M(T)) \in \prod_{i=1}^M \mathcal L_i$,
find all $P(X_1, \ldots, X_m) \in \F_q[X_1, \ldots, X_m]$ with $\deg(P) \leq d$
such that for each $i \in [M]$,
$$P\circ \gamma_{\fraka_i} (T) = Q_i(T).$$
(This is a system of linear equations over $\F_{q^m}$ with ${ d+m \choose m}$ variables and $(d+1) \cdot M$ constraints).
\item Output the list of all such $P(X_1, \ldots, X_m)$.
\end{enumerate}

Before analyzing correctness, let us comment on the running
time of this reduction.
For constant $m, s$, we claim that this reduction can be
implemented to run in near-linear time.
The running time of this reduction depends the size 
$\mathcal L_i$, and on how quickly
we can find $P(X_1, \ldots, X_m)$ given 
$Q_1(T), \ldots, Q_M(T)$. If we happen to know that
the $\mathcal L_i$ are all of constant size (as it happens
whenever the fraction of errors is bounded below the
Johnson radius), then the factor of $\prod | \mathcal L_i |$
is at most a constant.
Thus we only need to show that
the latter step can be implemented in near-linear time.
Given $Q_1(T), \ldots, Q_M(T)$, we can: (1) evaluate each
$Q_i(T)$ and each of its derivatives of order $<s$ at all
points of $\F_q^m$, (2) via Lemma~\ref{lem:multcompose}, we know that
this suffices to give us $P^{(<s)}(\ba)$ for each $\ba \in \F_q^m$,
(3) to interpolate $P(\vecX)$ given $P^{(<s)}(\ba)$ for 
each $\ba \in \F_q^m$.

Thus, we get that unique decoding of multivariate multiplicity codes
upto half the minimum distance can be done in near-linear time.
Furthermore, if Nielsen's algorithm for list-decoding univariate multiplicity
codes upto the Johnson bound can be implemented in near-linear time,
then one can also list-decode multivariate multiplicity codes upto (almost) the
Johnson bound in near-linear time.

\subsubsection{Analysis of the decoding algorithm}

Suppose $P(\vecX) \in \F_q[\vecX]$ is such that
$\Delta(\Enc_{s,d,m,q}(P), r) < \delta_0$.
Let $E \subseteq \F_q^m$ be the set of $\ba \in \F_q^m$
where $P^{(<s)}(\ba)$ differs from $r^{(<s})(\ba)$.

We first show that $Q_{i, \bl}(T) := P^{(\bl)} \circ \gamma_{\fraka_i}(T) \in \mathcal L_{i, \bl}$. Indeed, by Lemma~\ref{lem:multcompose}, for every 
$t \in \F_{q^m}$ such that $\gamma_{\fraka} (t) \not\in E$,
and for every $j < s-c+1$, we have:
\begin{align*}
Q_{i, \bl}^{(j)}(t)  &=
\sum_{\ii : \wt(\ii) = j} {\ii + \bl \choose \ii} P^{(\ii+\bl)}(\gamma_{\fraka_i}(t)) \fraka_i^{\ii}\\
&= \sum_{\ii : \wt(\ii) = j} {\ii + \bl \choose \ii} r^{(\ii+\bl)}(\gamma_{\fraka_i}(t)) \fraka_i^{\ii}\\
&= \sum_{\ii : \wt(\ii) = j} {\ii + \bl \choose \ii} r^{(\ii+\bl)}(\gamma_{\fraka_i}(t)) \fraka_i^{\ii}\\
&= \ell_{i, \bl}(t).
\end{align*}
Thus $\Delta(\Enc_{s-c+1,d,1,q}(Q_{i, \bl}), \ell_{i, \bl}) \leq \frac{|E|}{q^m} < \delta_0$, and thus $Q_{i, \bl}$ is indeed included in the list $\mathcal L_i$.
The crucial points here are  (a) the
relative distance of the univariate multiplicity code of
order $s-c+1$ evaluations of degree $d$ polynomials has minimum
distance $1 - \frac{1 - \delta}{1 - \gamma}$, and (b) our choice of $\gamma$ ensures that
$\delta_0 \leq \eta( 1 - \frac{1-\delta}{1-\gamma})$.
Thus algorithm $\mathcal A$ is indeed capable of finding $\mathcal L_i$ as 
required by the reduction.

\subsubsection{Algebraically repelling curves}
\label{sec:repel}

In this section we  prove Lemma~\ref{lem:multcompose} and Lemma~\ref{lem:zerozero}.

\begin{proofof}{Lemma \ref{lem:multcompose}}
By definition of derivatives, we have:
$$Q_{\bl}(t + W) = \sum_{j} Q_{\bl}^{(j)}(t) W^j,$$
$$P^{(\bl)}(\gamma_{\fraka}(t) + \vecX) = \sum_{\ii}  (P^{(\bl)})^{(\ii)}(\gamma_{\fraka}(t))\vecX^{\ii} = \sum_{\ii} {\bl + \ii \choose \ii}  P^{(\bl+\ii)}(\gamma_{\fraka}(t))\vecX^{\ii}.$$

%Take $t \in \F_{q^m}$ such that $\gamma_{\fraka}(t) = \ba$.
By linearity, $\gamma_{\fraka}(t + W) = \gamma_{\fraka}(t) + \gamma_{\fraka}(W)$.
So 
\begin{align*}
Q_{\bl}(t + W) &= P^{(\bl)}\circ \gamma_{\fraka}(t+W)\\
&= P^{(\bl)}(\gamma_{\fraka}(t) + \gamma_{\fraka}(W))\\
&= \sum_{\ii} {\bl + \ii \choose \ii} P^{(\bl+\ii)}(\gamma_{\fraka}(t))(\gamma_{\fraka}(W))^{\ii} \\
\end{align*}

%By assumption, $Q\circ \gamma_{\fraka}(t+W) = 0$.
%&= \sum_{j} \left(\sum_{\ii: \wt(\ii) = j} Q^{(\ii)}(\gamma_{\fraka}(t))\gamma_{\fraka}^{\ii}\right)(W))^{\ii} \\
Taking this equation mod $W^q$, we get the following equation:
\begin{align*}
\sum_{j < q} Q_{\bl}^{(j)}(t) W^j &= \sum_{\ii : \wt(\ii) < q} {\bl + \ii \choose \ii} P^{(\bl + \ii)}(\gamma_{\fraka}(t))(\fraka W)^{\ii}  \quad \mod W^q
%&= \sum_{j : 0 \leq j < q} \sum_{\ii : \wt(\ii) = j} Q^{(\ii)}(\gamma_{\fraka}(t))(\fraka^{\ii}$.
\end{align*}
For $j < q$, note that the coefficient of $W^j$ in the right hand side of this
equation equals $\sum_{\ii: \wt(\ii) = j} {\bl + \ii \choose \ii} P^{(\bl + \ii)}(\gamma_{\fraka}(t)) \fraka^{\ii}$.
On the other hand, the coefficient of $W^j$ in the left hand side of the equation equals $Q_{\bl}^{(j)}(t)$.
%$$ P(t + W)  = 
We therefore conclude that for each $j$ with $0 \leq j < q$, we have
$$Q_{\bl}^{(j)} (t) = \sum_{\ii: \wt(\ii) = j} {\bl + \ii \choose \ii} P^{(\bl + \ii)}(\gamma_{\fraka}(t))\fraka^{\ii}.$$
\end{proofof}

\begin{proofof}{Lemma~\ref{lem:zerozero}}
We will show that for each $\ba \in \F_{q}^m$, $\mult(P, \mathbf a) \geq s$.
Then by Lemma~\ref{lemma:schwartz} (recalling that $\deg(P) < sq$), we
can conclude that $P(X_1, \ldots, X_m) = 0$.

Fix $i \in [M]$ and $\bl$ with $\wt(\bl) < c$.
We have $P^{(\bl)} \circ \gamma_{\fraka_i}(T) = 0$.
By Lemma~\ref{lem:multcompose}, we conclude that for every $t \in \F_{q^m}$ and
$j < q$:
$$\sum_{\ii: \wt(\ii) = j}{\bl + \ii \choose \ii} P^{(\bl+\ii)}(\gamma_{\fraka_i}(t))\fraka_i^{\ii} = 0.$$

Thus for every $\ba \in \F_q^m$, every $i \in [M]$, $0 \leq j < s$, and $\bl$ with
$\wt(\bl) < c$:
\begin{align}
\label{eqvanzero}
\sum_{\ii: \wt(\ii) = j} {\bl + \ii \choose \ii} P^{(\bl+\ii)}(\ba)\fraka_i^{\ii} = 0.
\end{align}

%Fix $\ba = (a_1, \ldots, a_m) \in \F_q^m$.

%By definition of derivatives, we have:
%$$Q(\ba + X) = \sum Q^{(\ii)}(\ba)\vecX^{\ii}.$$

%Take $t \in \F_{q^m}$ such that $\gamma_{\fraka}(t) = \ba$.
%By definition, $\gamma_{\fraka}(t + W) = \gamma_{\fraka}(t) + \gamma_{\fraka}(W)$.
%So 
%\begin{align*}
%Q\circ \gamma_{\fraka}(t+W) &= Q(\gamma_{\fraka}(t) + \gamma_{\fraka}(W))\\
%&= \sum_{\ii} Q^{(\ii)}(\ba)(\gamma_{\fraka}(W))^{\ii}.
%\end{align*}

%By assumption, $Q\circ \gamma_{\fraka}(t+W) = 0$.
%Taking this equation mod $W^q$, we get
%$\sum_{\ii : \wt(\ii) < q} Q^{(\ii)}(\ba)(\fraka W)^{\ii} = 0$.
%Note that the coefficient of $W^j$ in  the left hand side of this
%equation equals $\sum_{\ii: \wt(\ii) = j} Q^{(\ii)}(\ba) \fraka^{\ii}$.
%We therefore conclude that for each $j$ with $0 \leq j < q$, we have
%$$\sum_{\ii: \wt(\ii) = j} Q^{(\ii)}(\ba)\fraka^{\ii} = 0.$$
%Suppose $\fraka$ is such that $Q \circ \gamma_{\fraka}(T) = 0$. 

%What we just showed can be summarized as follows.
For $0 \leq j' < q$, let $R_{\ba, j'}(\vecY) \in \F_q[\vecY]$ be the polynomial
$$ R_{\ba, j'}(\vecY) = \sum_{\ii' : \wt(\ii') = j'} P^{(\ii')}(\ba) \vecY^{\ii'}.$$
Then the derivatives of $R_{\ba, j'}$ are given by:
\begin{align*}
R_{\ba, j'}^{(\bl)}(\vecY)
&= \sum_{\ii':\wt(\ii') = j'} {\ii' \choose \bl} P^{(\ii')}(\ba) \vecY^{\ii' - \bl}\\
&= \sum_{\ii:\wt(\ii) = j'-\wt(\bl)} {\ii + \bl \choose \bl} P^{(\ii + \bl)}(\ba) \vecY^{\ii}\\
\end{align*}

Equation~\eqref{eqvanzero} says that for each $\ba \in \F_q^m$, $i \in [M]$, $\bl$ with $\wt(\bl) < c$, and $j' < q + \wt(\bl)$, 
we have
$$R_{\ba, j'}^{(\bl)}(\fraka_i) = 0.$$

Thus for every $j < q$ and $\ba, i, \bl$ as above:
$$\mult(R_{\ba, j}, \fraka_i) \geq c.$$

%Applying this with $\fraka$ being each of $\fraka_1, \ldots, \fraka_M$,
%we see that $R_{\ba, j}(\fraka_i)= 0$ for each $i \in M$.

By the general position hypothesis on $\fraka_i$, this implies that for each $j < s$, the
polynomial $R_{\ba, j}(\vecY)$ is itself identically $0$. 

But the coefficients of 
$R_{\ba, j}(\vecY)$ are $P^{(\ii)}(\ba)$, for $\ii$ satisfying $\wt(\ii) = j$. 
Thus $P^{(\ii)}(\ba) = 0$ for each $\ba$ and each $\ii$ with $\wt(\ii) < s$.

Therefore $\mult(P, \ba) \geq s$ for each $\ba \in \F_q^m$, which
implies that $P = 0$, as desired.
\end{proofof}

\subsection{Local List-Decoding}

The list-decoding algorithm for multivariate multiplicity codes
from the Johnson radius can be used to give a
{\em local list-decoding} algorithm for multivariate multiplicity
codes upto the Johnson radius. Since the definition of local
list-decoding is somewhat involved, we refer the reader to
the appendix of~\cite{K} for the algorithm and its analysis.

%The global list-decoding algorithms for multivariate multiplicity 
%codes can also be used to give {\em local} list-decoding
%algorithms for multivariate multiplicity codes. See~\cite{K}
%for details.

\section{Encoding}
\label{sec:encoding}

In this section we discuss encoding algorithms.

Since multiplicity codes are $\F_q$-linear subspaces
of $\Sigma^{\F_q^m}$, it is natural to choose an encoding map
which is $\F_q$-linear. One very natural encoding map to consider
is the map:
$$ E: \F_q^{d + m \choose m } \to \mathcal C,$$
which treats its input as a vector of coefficients of
monomials for a polynomial $P \in \F_q[\vecX]$ of degree at most $d$,
and outputs $\Enc_{s,d,m,q}(P)$. It is well known that the
task of computing $E$, namely evaluating a given polynomial
and all its derivatives of order at most $s$ at all points of $\F_q^m$
can be performed in near-linear time $O( (d^m + q^m) \cdot {m+s \choose m} \cdot \log(d^m + q^m))$.

However, for the purposes of local decoding, it will be important
to choose the encoding map $E: \Sigma_0^k \to \mathcal C \subseteq \Sigma^n$
a bit more carefully. The goal
is to have the encoding be {\em systematic}; i.e., to have the 
symbols of the message appear as symbols (or parts of symbols) 
of its encoding. Once we have such an encoding map, a local correction
algorithm immediately gives us a local decoding algorithm.

Such a systematic encoding map can be chosen by giving an 
{\em interpolating set}. Concretely, for a given $s, d, m, q$, we
want a set $S \subseteq \F_q^m \times \{ \ii \mid \wt(\ii) < s\}$
such that for every $f: S \rightarrow \F_q$, there is {\bf exactly one} $P(X_1, \ldots, X_m) \in \F_q[X_1, \ldots, X_m]$
of degree at most $d$ such that for each $(\ba, \ii) \in S$, $P^{(\ii)}(\ba) = f(\ba, \ii)$.

It is easy to see that such sets $S$ exist,
and any such set $S$ must have $|S| = {d + m \choose m}$. 
In order for the so obtained local decoding algorithm to
run in sublinear time (assuming the local correction algorithm runs
in sublinear time), it will be important that this interpolating
set be {\em explicit}: given message coordinate $i \in  [k]$, we should
be able to compute, in time $\poly(\log(k))$,
the codeword coordinate $j \in [n]$ which contains the $i$th symbol
of the codeword.
In~\cite{K}, an explicit such interpolating set was given.

\begin{theorem}[\cite{K}]
\label{thm:interp}
There exist explicit interpolating sets $S$ as above.

Thus there exist explicit systematic encoding maps for
multiplicity codes.
\end{theorem}

The interpolating sets for multiplicity codes constructed 
above are in fact simple combinations of interpolating sets
for Reed-Muller codes. Furthermore, 
it is known that there exist interpolating sets for Reed-Muller
codes from which polynomial interpolation can be performed in near-linear time.
This implies, by inspecting the proof of Theorem~\ref{thm:interp},
that the encoding map described above can be computed in near-linear time.

\section{Discussion}

\begin{enumerate}

\item The improved local decoding algorithm given in Section~\ref{sec:local}
had two main ideas over the original local decoding algorithm of~\cite{KSY}:
getting more information from each line, and robustly combining this information
across different lines by decoding a multiplicity code.

The first idea is naturally motivated by an incongruity between the three papers: \cite{KSY} (on local correction of multiplicity codes), \cite{GKS} (on affine invariant codes, local correction of affine invariant codes, and their relationship to bounds on Nikodym sets), and \cite{DKSS} (giving lower bounds on the size of Nikodym sets via the extended method of multiplicities).

A Nikodym set is a set $N \subseteq \F_q^m$ such that for every
$\ba \in \F_q^m$, there is some line $L \subseteq \F_q^m$ passing through
$\ba$, such that $L \setminus \{\ba\} \subseteq N$ (i.e., the entire
line, except possibly $\ba$, is contained in $N$).
In~\cite{GKS} it was noted that lower bounds on Nikodym sets 
follow from the existence of algebraic error-correcting codes that
can be locally corrected at $\ba \in \F_q^m$ by querying a line through
$\ba$. In~\cite{DKSS}, the extended method of multiplicities (which in
retrospect can be interpreted in the language of multiplicity codes)
was used to show a lower bound of $\left(\frac{q}{2}\right)^m$ on the
size of Nikodym sets. On the other hand, if one tried to use the original local
correction algorithm of~\cite{KSY} for multiplicity codes
to prove a lower bound on Nikodym sets (via the~\cite{GKS} connection),
we would get a significantly weaker bound than the bound of~\cite{DKSS}.
All this suggests that there should be a better algorithm for local
correction of multiplicity codes, and (looking at the details) a way to get more
information from each line.

The second idea is motivated by a different incongruity: the local correction
algorithm of~\cite{KSY} combines information from different lines
by decoding a Reed-Muller code, but there
seems to be no reason for the $s = 1$ case of multiplicity codes
(i.e., Reed-Muller codes) to
receive preferential treatment amongst all possibilities for $s$.
The new algorithm resolves this incongruity, and lets us use
fewer lines for the local decoding.

\item The decoding algorithm described in Section~\ref{sec:var} based on
$m$ linearly independent lines was motivated by another problem in the combinatorial
geometry: the joints problem~\cite{GK-joints, KSS}.
Alex Vardy and Abdul Basit independently suggested to me that bounds on the joints problem
could potentially be improved using the method of multiplicities. The $m$-line decoding algorithm
for multiplicity codes can be used to give a multiplicity-enhanced proof of the~\cite{KSS}
bound on the joints problem; unfortunately this variant does not improve on the bounds.

A brief outline of the argument goes as follows. One first interpolates a low degree
polynomial vanishing with multiplicity
 at least $a$ at the joints of interest (for some large $a$). Then one argues that this polynomial vanishes at each line
of the collection with multiplicity at least some $c$. Then using Lemma~\ref{lem:joints2}, we deduce
that this polynomial actually vanishes at each of the joints with multiplicity at least $c' \gg a$.
This leads to a contradiction, unless the number of joints is large.

\item The Reed-Muller codes over small fields and in
many variables have been very influential.
The analogous multiplicity
code would be based on order $s$ evaluations of polynomials over
$\F_q$ of total degree at most $d$ and individual degree at most
$sq- 1$. It would
be very interesting to study properties of these codes, and to see
if they have any coding/combinatorial applications.

We note a curious example of how things change when the individual degrees
are bounded: a polynomial $P(X_1, \ldots, X_m) \in \F_2[X_1, \ldots, X_m]$
of individual degree at most $1$ vanishes at a point $\ba \in \F_2^m$ with
multiplicity at least $s$
if and only if $P$ vanishes at all points of $\F_2^n$ which 
are at a Hamming distance at most $s-1$ from $\ba$.

\item There are various interesting variations/cousins
of Reed-Solomon and Reed-Muller codes which are interesting from the
coding perspective.
These include algebraic-geometric codes, BCH codes,
projective Reed-Muller codes, Grassman codes, etc. It would
be interesting to investigate multiplicity-based generalizations
of these codes.
For algebraic-geometric codes, some such investigations were made 
(in the 1-dimensional case, analogous to Reed-Solomon codes) 
by Rosenbloom-Tsfasman~\cite{RT-m-metric}, Xing~\cite{Xing}
and Nielsen~\cite{Nielsen}.

\end{enumerate}

\section{Open Questions}

We conclude with a list of some interesting open questions.
Some of these questions are open even in the classical $s=1$ case
(Reed-Solomon codes and Reed-Muller codes).

\begin{enumerate}

\item What is the list-decoding radius for univariate multiplicity codes?
In other words, what is the largest fraction $\eta$ of errors from which
univariate multiplicity codes of distance $\delta$ and
block-length $n$ can be list-decoded with 
$\poly(n)$ list-size? For general univariate multiplicity codes,
this is only known to be true for $\eta \leq 1 - \sqrt{1-\delta}$,
while for univariate multiplicity codes over prime fields with
$s$ sufficiently large, it is known for every $\eta < \delta$.

This question is even open for Reed-Solomon codes.

Here are some related questions. Does the answer depend on the field?
Does the answer depend on the set of evaluation points? What happens
for multivariate multiplicity codes (again, this is open for multivariate
Reed-Muller codes too)?

\item It is an extremely interesting question whether the list-size for 
list-decoding univariate multiplicity codes over prime fields (as in~\cite{GW, K}) needs to be $\poly(q)$, or if it can be reduced to a constant independent
of $q$.

It is also extremely interesting to know whether the primality of
$q$ is essential for the improved list-decoding of multiplicity codes
in~\cite{GW, K}.

\item Can Nielsen's algorithm for
list-decoding univariate multiplicity codes upto the Johnson bound
be implemented to run
in near-linear time. It would also immediately imply a near-linear time
global algorithm, and a faster local algorithm, for list-decoding multivariate multiplicity codes upto the Johnson radius.
This seems to require some nontrivial adaptation of the ideas
of Alekhnovich~\cite{Alekhnovich}, who showed how to list-decode Reed-Solomon codes in near-linear time.

\item Are multiplicity codes locally testable?

\item Are there any applications of multiplicity codes to computational complexity theory? Reed-Solomon codes and Reed-Muller codes have found many celebrated
applications, and it would be interesting to see if multiplicity codes
can improve on any of them.

\item Are there any practical applications of multiplicity codes? By combining
high rate with sublinear-time decoding, and also supporting various other efficent operations, multiplicity codes seem to be theoretically practical. Perhaps they really are practical?

\end{enumerate}

\section*{Acknowledgements}

Thanks to Alexander Barg and Oleg Musin for organizing an excellent
workshop to celebrate Ilya Dumer's birthday, and for encouraging me to
write this article. Many many thanks to Ilya Dumer for being
a wonderful friend and mentor for many years, and for getting me hooked on
coding theory in the first place.

\bibliography{multcode-survey}
\bibliographystyle{alpha}

\end{document}